\documentclass[prd,onecolumn,nopacs,nofootinbib]{revtex4}

\usepackage{amsmath}
\usepackage{graphicx}
\usepackage{amsfonts}
\usepackage{latexsym}
\usepackage{bbold}
\usepackage{wasysym}
\usepackage{calligra}
\usepackage{float}
\usepackage{ulem}
\usepackage{inputenc}
\usepackage{xspace}
\usepackage{url}
\usepackage{epstopdf}
\usepackage{tikz}
\usepackage{amsthm}

\usepackage[framemethod=TikZ]{mdframed}
\usepackage{lipsum}

\newcommand{\be}{\begin{equation}}
\newcommand{\ee}{\end{equation}}
\newcommand{\beq} {\begin{equation}}
\newcommand{\eeq} {\end{equation}}
\newcommand{\ba}{\begin{eqnarray}}
\newcommand{\ea}{\end{eqnarray}}

\newtheorem{theo}{Theorem}

\begin{document}

	\title{Solving Linear Tensor Equations}
	%\title{Scale-invariant and projective-invariant theories in metric-affine geometry}
	
	\author{Damianos Iosifidis}
	\affiliation{Institute of Theoretical Physics, Department of Physics
		Aristotle University of Thessaloniki, 54124 Thessaloniki, Greece}
	\email{diosifid@auth.gr}
	
	\date{\today}
	\begin{abstract}
		
		We develop a systematic way to solve linear equations involving tensors of arbitrary rank. We start off with the case of a rank $3$ tensor, which appears in many applications, and after finding the condition for a unique solution we derive this solution. Subsequently we generalize our result to tensors of arbitrary rank. Finally we consider a generalized version of the former case of rank $3$ tensors and  extend the result when the tensor traces  are also included.

	\end{abstract}
	
	\maketitle
	
	\allowdisplaybreaks
	
	%\newpage
	
%	\tableofcontents
	%\newpage
	
	\section{Introduction}
	In many applications a tensorial equation of the form
		\beq
		a_{1}N_{\alpha\mu\nu}+a_{2}N_{\nu\alpha\mu}+a_{3}N_{\mu\nu\alpha}+a_{4}N_{\alpha\nu\mu}+a_{5}N_{\nu\mu\alpha}+a_{6}N_{\mu\alpha\nu}=B_{\alpha\mu\nu} \label{eq10}
\eeq
appears, where $B_{\alpha\mu\nu}$ is some given (i.e. known)	tensor field, $a_{i}'s$, $i=1,2,...,6$ are some given scalar fields and $N_{\alpha\mu\nu}$ is the unknown tensor field one wishes to solve for. For instance this tensorial equation is encountered when one varies the quadratic Metric-Affine Gravity action  \cite{iosifidis2019metric,iosifidis2019exactly,hehl1995metric} with respect to the affine connection\footnote{Third order tensors appear also in mechanics, see for instance \cite{qi2018third}.}. There, $B_{\alpha\mu\nu}$ represents the (known) hypermomentum source and $N_{\alpha\mu\nu}$ is the distortion tensor \cite{schouten1954ricci} in which spacetime torsion and non-metricity are encoded. Then having solved for $N_{\alpha\mu\nu}$ entirely in terms of the sources, that is combinations of $B_{\alpha\mu\nu}$ one can easily obtain the forms of torsion and non-metricity, namely the non-Riemannian \cite{eisenhart2012non} parts of the geometry. In order to solve this equation one could go about and split $N_{\alpha\mu\nu}$ into its irreducible decomposition and then take contractions, symmetrizations etc. in order to find the various pieces in terms of $B_{\alpha\mu\nu}$ and its contractions. Even though this may work in some cases, it will be a difficult task in general. Moreover, this procedure will fall short quickly if one wishes to generalize the above considerations and ask for the general solution ($N$ in terms of $B$) of the rank-n tensorial equation
		\beq
\underbrace{	a_{1}N_{\mu_{1}\mu_{2}\mu_{3}...\mu_{n}}+a_{2}N_{\mu_{n}\mu_{1}\mu_{2}...\mu_{n-1}}+...+a_{n}N_{\mu_{2},\mu_{3},...\mu_{n}\mu_{1}}+....}_{n!-terms}=B_{\mu_{1}\mu_{2}\mu_{3}...\mu_{n}} \label{eqn}
	\eeq
	\label{intro}
 Evidently, one easily realizes that it would be impossible to solve the latter by resorting to some decomposition scheme for $N$\footnote{For decompositions of  rank-$3$ tensors see \cite{itin2020decomposition} and for there geometric picture one may consult \cite{auffray2013geometrical}. In addition, a nice review on tensor calculus can be found in \cite{landsberg2012tensors}. }. It is then natural to ask, is there a systematic and practical way to solve equations of the form $(\ref{eq10})$ or more generally ($\ref{eqn}$). It is the purpose of this letter to answer this question. As we show, under a fairly general  non-degeneracy condition it is always possible to find the unique solution of ($\ref{eq10}$), or more generally of ($\ref{eqn}$), by following a certain procedure that we develop below along with some extensions/generalizations.

	\section{The Theorems}
	In what follows we present $3$ Theorems.In the first one, the systematic way to solve equation ($\ref{eq10}$) for $N$ is proved. We then extend this result to tensors $N$ of arbitrary rank (i.e not necessary $3$) and solve equations of the form ($\ref{eqn}$). Finally we derive the solution of a generalized version of $(\ref{eq10})$ where the traces of $N$ are also included. We have the following.

	\begin{theo}
		Consider the tensor equation
		\beq
		a_{1}N_{\alpha\mu\nu}+a_{2}N_{\nu\alpha\mu}+a_{3}N_{\mu\nu\alpha}+a_{4}N_{\alpha\nu\mu}+a_{5}N_{\nu\mu\alpha}+a_{6}N_{\mu\alpha\nu}=B_{\alpha\mu\nu} \label{eq1}
				\eeq
				where  $a_{i}$, $i=1,2,...,6$ are scalars, $B_{\alpha\mu\nu}$   is a given (known) tensor and $N_{\alpha\mu\nu}$ are the components of the unknown tensor\footnote{Of course the result holds true even when $N_{\alpha\mu\nu}$ are the components of a tensor density instead or even of a connection given that $B_{\alpha\mu\nu}$ are also  of the same kind. } $N$. Define the matrix
				 \begin{equation}
				A := 
				\begin{pmatrix}
				a_{1} & a_{2} & a_{3} & a_{4} & a_{5} & a_{6} \\
				a_{2} & a_{3} &  a_{1} & a_{5} & a_{6} & a_{4} \\
				a_{3} & a_{1} &  a_{2} & a_{6} & a_{4} & a_{5} \\
				a_{6} & a_{5} &  a_{4} & a_{3} & a_{2} & a_{1} \\
				a_{5} & a_{4} &  a_{6} & a_{2} & a_{1} & a_{3} \\
				a_{4} & a_{6} &  a_{5} & a_{1} & a_{3} & a_{2} \\
				\end{pmatrix} \label{A}
				\end{equation}
				If the system is non-degenerate, that is if	
				 \footnote{ A necessary condition for this to happen is that	$
				 	\sum_{i=1}^{6}a_{i}\neq 0 	$. However, this condition alone is not sufficient since the latter quantity can be non-vanishing but it may be so that the full determinant still vanishes. See Appendix for more details on this feature.}
				\beq
				det(A) \neq 0
				\eeq
		holds true, then the general and unique solution of ($\ref{eq1}$) reads
				\beq
			N_{\alpha\mu\nu}=\tilde{a}_{11}	B_{\alpha\mu\nu}+\tilde{a}_{12}B_{\nu\alpha\mu}+\tilde{a}_{13}B_{\mu\nu\alpha}+\tilde{a}_{14}B_{\alpha\nu\mu}
			+\tilde{a}_{15} B_{\nu\mu\alpha}+\tilde{a}_{16} B_{\mu\alpha\nu} \label{co}
			\eeq
				where the $\tilde{a}_{1i}$ $  's$ are the first row elements of the inverse matrix $A^{-1}$.
	\end{theo}

\begin{proof}
	Starting from ($\ref{eq1}$) we perform the $5$ independent possible permutations on the indices and including also $($\ref{eq1}$)$ we end up with the system
		\begin{gather}
	a_{1}N_{\alpha\mu\nu}+a_{2}N_{\nu\alpha\mu}+a_{3}N_{\mu\nu\alpha}+a_{4}N_{\alpha\nu\mu}+a_{5}N_{\nu\mu\alpha}+a_{6}N_{\mu\alpha\nu}=B_{\alpha\mu\nu}  \nonumber \\
	a_{1}N_{\nu\alpha\mu}+a_{2}N_{\mu\nu\alpha}+a_{3}N_{\alpha\mu\nu}+a_{4}N_{\mu\alpha\nu}+a_{5}N_{\alpha\nu\mu}+a_{6}N_{\nu\mu\alpha}=B_{\nu\alpha\mu} \\ \nonumber
		a_{1}N_{\mu\nu\alpha}+a_{2}N_{\alpha\mu\nu}+a_{3}N_{\nu\alpha\mu}+a_{4}N_{\nu\mu\alpha}+a_{5}N_{\mu\alpha\nu}+a_{6}N_{\alpha\nu\mu}=B_{\mu\nu\alpha}  \\  \nonumber
		a_{1}N_{\alpha\nu\mu}+a_{2}N_{\mu\alpha\nu}+a_{3}N_{\nu\mu\alpha}+a_{4}N_{\alpha\mu\nu}+a_{5}N_{\mu\nu\alpha}+a_{6}N_{\nu\alpha\mu}=B_{\alpha\nu\mu} \\ \nonumber
		a_{1}N_{\nu\mu\alpha}+a_{2}N_{\alpha\nu\mu}+a_{3}N_{\mu\alpha\nu}+a_{4}N_{\nu\alpha\mu}+a_{5}N_{\alpha\mu\nu}+a_{6}N_{\mu\nu\alpha}=B_{\nu\mu\alpha} \\ \nonumber
		a_{1}N_{\mu\alpha\nu}+a_{2}N_{\nu\mu\alpha}+a_{3}N_{\alpha\nu\mu}+a_{4}N_{\mu\nu\alpha}+a_{5}N_{\nu\alpha\mu}+a_{6}N_{\alpha\mu\nu}=B_{\mu\alpha\nu} 
	\end{gather}

	Then, defining the  matrix
 \begin{equation}
A := 
\begin{pmatrix}
a_{1} & a_{2} & a_{3} & a_{4} & a_{5} & a_{6} \\
a_{2} & a_{3} &  a_{1} & a_{5} & a_{6} & a_{4} \\
a_{3} & a_{1} &  a_{2} & a_{6} & a_{4} & a_{5} \\
a_{6} & a_{5} &  a_{4} & a_{3} & a_{2} & a_{1} \\
a_{5} & a_{4} &  a_{6} & a_{2} & a_{1} & a_{3} \\
a_{4} & a_{6} &  a_{5} & a_{1} & a_{3} & a_{2} \\
\end{pmatrix} 
\end{equation}
	along with the columns
	 \beq
	 \mathcal{N}:=(N_{\alpha\mu\nu},N_{\nu\alpha\mu},N_{\mu\nu\alpha},N_{\alpha\nu\mu},N_{\nu\mu\alpha},N_{\mu\alpha\nu})^{T}
	 \eeq
	  and 
	  \beq
	  \mathcal{B}:=(B_{\alpha\mu\nu},B_{\nu\alpha\mu},B_{\mu\nu\alpha},B_{\alpha\nu\mu},B_{\nu\mu\alpha},B_{\mu\alpha\nu})^{T}
	  \eeq
	we may express the above system in matrix form as
	\beq
	A \mathcal{N}=\mathcal{B}
	\eeq 
	In the above $\mathcal{N}$ is the column consisting of the unknown elements we wish to find. Since, by hypothesis, we have a non-degenerate system it follows that $det(A)\neq 0$ and as a result the inverse $A^{-1}$ exists. We then, formally multiply the above equation by $A^{-1}$ from the left to get
	\beq
		 \mathcal{N}=A^{-1} \mathcal{B}
	\eeq
	The above is a column equation and of course each element on the left column must be equal to each element on the right. Equating the first element we arrive at the stated result 
		\beq
	N_{\alpha\mu\nu}=\tilde{a}_{11}	B_{\alpha\mu\nu}+\tilde{a}_{12}B_{\nu\alpha\mu}+\tilde{a}_{13}B_{\mu\nu\alpha}+\tilde{a}_{14}B_{\alpha\nu\mu}
	+\tilde{a}_{15} B_{\nu\mu\alpha}+\tilde{a}_{16} B_{\mu\alpha\nu} \label{co}
	\eeq
	where the $\tilde{a}_{1i}$ $  's$ are the elements of the first row of the inverse matrix $A^{-1}$ which, of course, depend on $a_{1}, a_{2},...,a_{6}$. Note that the equations we get for the rest of the column elements will be related to the above one with cyclic permutations and will therefore give nothing new. Concluding, $(\ref{co})$ is the general solution of ($\ref{eq1}$).
	Some comments are now in order.

	\textbf{Comment 1.}  Note that if $B_{\alpha\mu\nu}=0$ and the matrix $A$ is non-singular we have that $N_{\alpha\mu\nu}=0$ as a unique solution. It should be emphasized that the demand that $det(A)\neq 0$ is all essential in order for the full $N_{\alpha\mu\nu}$ tensor field to be vanishing. If the last requirement is not fulfilled the full $N$ tensor may as well not be identically vanishing since in this case not the full $N$ but certain (anti)-symmetrizations of it appear in $(\ref{eq1})$. In such an occasion only certain parts of $N_{\alpha\mu\nu}$ will be vanishing.
	
		\textbf{Comment 2.} If the components  $N_{\alpha\mu\nu}$ are symmetric or antisymmetric in any pair of indices then the system is greatly simplified and the $6\times 6$ matrix $A$ is reduced to a $3\times 3$ matrix instead.\footnote{This is easily realized as follows. Without loss of generality let us suppose that $N$ is symmetric in its first two indices, i.e. $N_{\alpha\mu\nu}=N_{\mu\alpha\nu}$. Then with this relation and circle permutations of it is is trivial to see that only three combinations of $N$ appear in $(\ref{eq1})$ and as a result the system reduces to a $3 \times 3$. Of course same goes also when $N$ is antisymmetric in any pair of its indices.}
	
	\end{proof}
		\begin{theo}
	In a $d-$dimensional space,	consider the tensor equation (with $n \le d$)
		\beq
\underbrace{	a_{1}N_{\mu_{1}\mu_{2}\mu_{3}...\mu_{n}}+a_{2}N_{\mu_{n}\mu_{1}\mu_{2}...\mu_{n-1}}+...+a_{n}N_{\mu_{2},\mu_{3},...\mu_{n}\mu_{1}}+....}_{n!-terms}=B_{\mu_{1}\mu_{2}\mu_{3}...\mu_{n}} \label{eq2}
	\eeq
		where  $a_{i}$, $i=1,2,...,n$ are scalars, $B_{\mu_{1}\mu_{2}...\mu_{n}} $   are the components of a  given (known) tensor and $N_{\mu_{1}\mu_{2}\mu_{3}...\mu_{n}}$ are the components of the unknown tensor $N$ of rank $n$.
		Define 
		 the square $n!\times n!$ matrix
		\begin{equation}
		A = 
		\begin{pmatrix}
		a_{1} & a_{2} & \cdots & a_{n!} \\
		a_{2} & a_{n} & \cdots & a_{n+1} \\
		\vdots  & \vdots  & \ddots & \vdots  \\
		a_{n+1} & a_{n!} & \cdots & a_{2} 
		\end{pmatrix}
		\end{equation}
		 Given that the system is non-degenerate, that is 
		$\det{A}\neq 0$, then the general and unique solution of ($\ref{eq2}$) is given by
		\beq
	N_{\mu_{1}\mu_{2}\mu_{3}...\mu_{n}}=\tilde{a}_{11}	B_{\mu_{1}\mu_{2}\mu_{3}...\mu_{n}}+...+\tilde{a}_{1n}B_{\mu_{2}\mu_{1}\mu_{3}...\mu_{n}}
	\eeq 
	where the $\tilde{a}_{1i}'s$ are the first row elements of the inverse matrix $A^{-1}$.
	\end{theo}
	
	\begin{proof}
	In an identical manner to the proof of Theorem $1$ we now start from ($\ref{eq2}$) and perform the $(n!-1)$ possible independent permutations to end up with the system of $n!$ equations\footnote{The first one is eq. ($\ref{eq2}$) itself.}
			\begin{gather}
		a_{1}N_{\mu_{1}\mu_{2}\mu_{3}...\mu_{n}}+a_{2}N_{\mu_{n}\mu_{1}\mu_{2}...\mu_{n-1}}+...+a_{n}N_{\mu_{2},\mu_{3},...\mu_{n}\mu_{1}}+....=B_{\mu_{1}\mu_{2}\mu_{3}...\mu_{n}}
\nonumber	\\
		a_{1}N_{\mu_{2}\mu_{3}\mu_{4}...\mu_{1}}+a_{2}N_{\mu_{1}\mu_{2}\mu_{3}...\mu_{n}}+...+a_{n}N_{\mu_{3}\mu_{4}...\mu_{1}\mu_{2}}+....=B_{\mu_{2}\mu_{3}\mu_{4}...\mu_{1}} \label{eq3}\nonumber \\
		... \\ \nonumber
				... \\  \nonumber
			...
					 \\ \nonumber
					 ... 
					 \nonumber \\
					(all-possible-permutations)
		\end{gather}
		
		We then define the square $n!\times n!$ matrix
		\begin{equation}
	A = 
	\begin{pmatrix}
	a_{1} & a_{2} & \cdots & a_{n!} \\
	a_{2} & a_{n} & \cdots & a_{n+1} \\
	\vdots  & \vdots  & \ddots & \vdots  \\
	a_{n+1} & a_{n!} & \cdots & a_{2} 
	\end{pmatrix}
	\end{equation}
		and the columns
		\beq
	\mathcal{N}:=(N_{\mu_{1}\mu_{2}\mu_{3}...\mu_{n}},N_{\mu_{n}\mu_{1}\mu_{2}...\mu_{n-1}},....)^T
		\eeq
		\beq
		\mathcal{B}:=(B_{\mu_{1}\mu_{2}\mu_{3}...\mu_{n}},B_{\mu_{n}\mu_{1}\mu_{2}...\mu_{n-1}},....)^T
		\eeq
		we may express the above system in the matrix form
		\beq
		A \mathcal{N}=\mathcal{B}
		\eeq 
	As in Theorem $1$ we then formally multiply the above equation by $A^{-1}$ from the left to get
		\beq
		\mathcal{N}=A^{-1} \mathcal{B}
		\eeq	
	and by equating the first row  element of the left and right hand sides of the above we arrive at the stated result
	\beq
	N_{\mu_{1}\mu_{2}\mu_{3}...\mu_{n}}=\tilde{a}_{11}	B_{\mu_{1}\mu_{2}\mu_{3}...\mu_{n}}+\tilde{a}_{12} B_{\mu_{n}\mu_{1}...\mu_{n-1}}...+\tilde{a}_{1n}B_{\mu_{2}\mu_{1}\mu_{3}...\mu_{n}}
	\eeq 
		where $\tilde{a}_{11},\tilde{a}_{11},...,\tilde{a}_{1n} $ are the  elements of the first row of the inverse matrix $A^{-1}$. 
		
\end{proof}
	\textbf{Remark.} Again, if the tensor $N$ has some symmetry property over some pair(s) of its indices the dimension of the matrix $A$ will be lowered accordingly.

	Now, going back to the case of a rank $3$ tensor one may ask how does the situation change when the traces of $N_{\alpha\mu\nu}$ also appear in $(\ref{eq1})$. Defining the three traces
	\beq
	N^{(1)}_{\mu}:=N_{\alpha\beta\mu}g^{\alpha\beta}\;\;, \;\; N^{(2)}_{\mu}:=N_{\alpha\mu\beta}g^{\alpha\beta}\;\;, \;\; N^{(3)}_{\mu}:=N_{\mu\alpha\beta}g^{\alpha\beta}
	\eeq
	the generalized version of ($\ref{eq1}$), still linear in $N$, including the above traces reads
		\beq
	a_{1}N_{\alpha\mu\nu}+a_{2}N_{\nu\alpha\mu}+a_{3}N_{\mu\nu\alpha}+a_{4}N_{\alpha\nu\mu}+a_{5}N_{\nu\mu\alpha}+a_{6}N_{\mu\alpha\nu}+\sum_{i=1}^{3}\Big( a_{7i}N^{(i)}_{\mu}g_{\alpha\nu}+a_{8i}N^{(i)}_{\nu}g_{\alpha\mu}+a_{9i}N^{(i)}_{\alpha}g_{\mu\nu} \Big)=B_{\alpha\mu\nu} \label{eq4}
	\eeq
	As we show below the appearance of the these extra terms does not introduce any serious technical difficulty and one can always solve for $N_{\alpha\mu\nu}$ in terms of a modified version of $B_{\alpha\mu\nu}$ the includes its traces. We have the following result.

	\begin{theo}
	Consider the $15$ parameter linear tensor equation ($\ref{eq4}$) where $N_{\alpha\mu\nu}$ are the components of the unknown tensor field. Define the matrices\footnote{The elements $\gamma_{ij}$ are linear combinations of the parameters $a_{i}$ and their exact relations are given in the appendix.}
	\begin{equation}
	\Gamma := 
	\begin{pmatrix}
	\gamma_{11} & \gamma_{12} & \gamma_{13}  \\
	\gamma_{21} & \gamma_{22} &  \gamma_{23} \\
	\gamma_{31} & \gamma_{32} &  \gamma_{33}  \\
	\end{pmatrix}
	\end{equation}
	and
	\begin{equation}
	A := 
	\begin{pmatrix}
	a_{1} & a_{2} & a_{3} & a_{4} & a_{5} & a_{6} \\
	a_{3} & a_{1} &  a_{2} & a_{5} & a_{6} & a_{4} \\
	a_{2} & a_{3} &  a_{1} & a_{6} & a_{4} & a_{5} \\
	a_{4} & a_{6} &  a_{5} & a_{1} & a_{3} & a_{2} \\
	a_{5} & a_{4} &  a_{6} & a_{2} & a_{1} & a_{3} \\
	a_{6} & a_{5} &  a_{4} & a_{3} & a_{2} & a_{1} \\
	\end{pmatrix}
	\end{equation}
	Then, given that both of the above matrices are non-singular, the unique  solution to $(\ref{eq4})$
	reads
		\beq
	N_{\alpha\mu\nu}=\tilde{a}_{11}	B_{\alpha\mu\nu}+\tilde{a}_{12}\hat{B}_{\nu\alpha\mu}+\tilde{a}_{13}\hat{B}_{\mu\nu\alpha}+\tilde{a}_{14}\hat{B}_{\alpha\nu\mu}
	+\tilde{a}_{15} \hat{B}_{\nu\mu\alpha}+\tilde{a}_{16} \hat{B}_{\mu\alpha\nu} \label{co2}
	\eeq
		where 
	\beq
	\hat{B}_{\alpha\mu\nu}={B}_{\alpha\mu\nu}-\sum_{i=1}^{3}\sum_{j=1}^{3}\Big( a_{7i}\tilde{\gamma}_{ij}B^{(j)}_{\mu}g_{\alpha\nu}+a_{8i}\tilde{\gamma}_{ij}B^{(j)}_{\nu}g_{\alpha\mu}+a_{9i}\tilde{\gamma}_{ij}B^{(j)}_{\alpha}g_{\mu\nu}\Big) \label{Bhat}
	\eeq
	
	\end{theo}
	\begin{proof}
		We begin by tracing equation ($\ref{eq4}$) three times independently with $g^{\alpha\mu}$, $g^{\alpha\nu}$ and $g^{\mu\nu}$ to get
		\begin{gather}
	\sum_{i=1}^{3}	\gamma_{1i}N_{\mu}^{(i)}=B_{\mu}^{(1)}\;\;, \;\;\sum_{i=1}^{3}	\gamma_{2i}N_{\mu}^{(i)}=B_{\mu}^{(2)} \;\;, \;\; \sum_{i=1}^{3}	\gamma_{3i}N_{\mu}^{(i)}=B_{\mu}^{(3)}
		\end{gather}

		where we have renamed all free indices to $\mu$ and the $\gamma_{ij}'s$ are some shorthand notations for certain linear combinations of the $a_{i}'s$ whose relations are given in the appendix. Now defining  the matrix $\Gamma$ with coefficients the $\gamma_{ij}'s$ along with the columns $\eta=(N^{(1)}_{\mu},N^{(2)}_{\mu},N^{(3)}_{\mu})^{T}$ and $b=(B^{(1)}_{\mu},B^{(2)}_{\mu},B^{(3)}_{\mu})^{T}$ we may express the above system in matrix form as
		\beq
		\Gamma n=b
		\eeq
		By hypothesis, the matrix $\Gamma$ is non-singular (i.e. $det(\Gamma)\neq0$) and therefore the inverse $\Gamma^{-1}$ exists and we may formally solve for $\eta$ as
		\beq
		\eta=\Gamma^{-1}b
		\eeq
		which in component notation translates to
		\beq
		N^{(i)}_{\mu}=\sum_{j=1}^{3}\tilde{\gamma}_{ij}B^{(i)}_{\mu}
		\eeq
		with the $\tilde{\gamma}_{ij}'s$ being the elements of $\Gamma^{-1}$.
		Then, substituting these last relations back in ($\ref{eq4}$), we fully eliminate the $N$ traces in favour of the traces of $B$, ending up with
			\beq
		a_{1}N_{\alpha\mu\nu}+a_{2}N_{\nu\alpha\mu}+a_{3}N_{\mu\nu\alpha}+a_{4}N_{\alpha\nu\mu}+a_{5}N_{\nu\mu\alpha}+a_{6}N_{\mu\alpha\nu}=\hat{B}_{\alpha\mu\nu} \label{eq5}
		\eeq
		where 
		\beq
		\hat{B}_{\alpha\mu\nu}={B}_{\alpha\mu\nu}-\sum_{i=1}^{3}\sum_{j=1}^{3}\Big( a_{7i}\tilde{\gamma}_{ij}B^{(j)}_{\mu}g_{\alpha\nu}+a_{8i}\tilde{\gamma}_{ij}B^{(j)}_{\nu}g_{\alpha\mu}+a_{9i}\tilde{\gamma}_{ij}B^{(j)}_{\alpha}g_{\mu\nu}\Big) \label{Bhat}
		\eeq
	We are pretty much done now since we can apply the result of Theorem $1$ to the modified tensor field $\hat{B}_{\alpha\mu\nu}$ in place of $B_{\alpha\mu\nu}$, completing therefore the proof
	\beq
	N_{\alpha\mu\nu}=\tilde{a}_{11}	\hat{B}_{\alpha\mu\nu}+\tilde{a}_{12}\hat{B}_{\nu\alpha\mu}+\tilde{a}_{13}\hat{B}_{\mu\nu\alpha}+\tilde{a}_{14}\hat{B}_{\alpha\nu\mu}
	+\tilde{a}_{15} \hat{B}_{\nu\mu\alpha}+\tilde{a}_{16} \hat{B}_{\mu\alpha\nu} 
	\eeq
where the modified components $\hat{B}_{\alpha\mu\nu}$ are given by ($\ref{Bhat}$).
		
		\end{proof}
	\section{Conclusions}
	
	We have formulated an analytical method that allows one to solve  tensorial equations of the form $(\ref{eq10})$, for the unknown tensor components $N_{\alpha\mu\nu}$ of the rank-3 tensor $N$. In particular we proved that under a fairly general non-degeneracy condition (i.e. $\det(A)\neq 0$) one can always solve equations of the form $(\ref{eq10})$ by simply finding the inverse of the matrix $A$ which is built from the coefficients $a_{i}$ appearing in the same equation. Subsequently, we generalized our result for arbitrary rank tensors and similarly obtained the solution of $\ref{eq2}$. Finally we extended the result we obtained for the first case (equation ($\ref{eq10}$)) to the $15$ parameter linear tensor equation ($\ref{eq4}$) including also the traces of $N_{\alpha\mu\nu}$ and obtained the unique solution for this case as well.

	As we already mentioned in the introduction, these results find a natural application in geometric extensions of General Relativity that take into account the non-Riemannian structure of spacetime (torsion and non-metricity). In this context of Metric-Affine Theories of Gravitation, equations of the form $(\ref{eq10})$ or more generally $(\ref{eq4})$ relate the distortion tensor with the hypermomentum (source). Therefore, the technique we developed here will be proven to be essential for finding how the matter sources produce spacetime torsion and non-metricity. This last  point is under consideration now.
	
	\section{Acknowledgments}	This research is co-financed by Greece and the European Union (European Social Fund- ESF) through the
	Operational Programme 'Human Resources Development, Education and Lifelong Learning' in the context
	of the project “Reinforcement of Postdoctoral Researchers - 2
	nd Cycle” (MIS-5033021), implemented by the
	State Scholarships Foundation (IKY).

		\appendix
		
		\section{The $\gamma's$}
		The relations between the elements of $\Gamma$ and the parameters $a_{i}$ read
		\begin{gather}
		\gamma_{11}=a_{1}+a_{3}+a_{71}+n a_{81}+a_{91}\;\;, \;\; \gamma_{12}=a_{2}+a_{4}+a_{72}+n a_{82}+ a_{92}\;\;, \;\; \gamma_{13}=a_{5}+a_{6}+a_{73}+n a_{83}+a_{93} \nonumber \\
			\gamma_{21}=a_{2}+a_{5}+n a_{71}+ a_{81}+a_{91}\;\;, \;\; \gamma_{22}=a_{1}+a_{6}+n a_{72}+a_{82}+ a_{92}\;\;, \;\; \gamma_{23}=a_{3}+a_{4}+n a_{73}+ a_{83}+a_{93} \nonumber \\		
			\gamma_{31}=a_{5}+a_{6}+a_{71}+ a_{81}+ n a_{91}\;\;, \;\; \gamma_{32}=a_{3}+a_{4}+a_{72}+ a_{82}+n a_{92}\;\;, \;\; \gamma_{31}=a_{1}+a_{2}+a_{73}+ a_{83}+n a_{93} \nonumber
		\end{gather}
		
		\section{The determinant of $A$}
		For the matrix $A$ as given by ($\ref{A}$), after some factorizations, its determinant is found to be (the use of Wolfram Mathematica \cite{wolfram1991mathematica} makes things easier here)
		\beq
		det(A)=\sigma_{1}\sigma_{2}\sigma_{3}^{2}
		\eeq
		with
		\beq
		\sigma_{1}=\sum_{i=1}^{6}a_{i}
		\eeq
		\beq
		\sigma_{2}=\sum_{i=1}^{3}(a_{i}-a_{i+3})
		\eeq
		\beq
		\sigma_{3}=\sum_{i=1}^{3}(a^{2}_{i}-a^{2}_{i+3})- \sum_{i=1}^{3}\sum_{j=1, \; j>i }^{3}(a_{i}a_{j}-a_{i+3}a_{j+3})
		\eeq
		Note that the determinant is of $6^{th}$ order on the $a_{i}'s$ as expected. Now, in order for the matrix to be non-singular all four sums above must be non-zero at the same time. If one (or more) of those sums is zero, this  implies that in ($\ref{eq1}$) a certain (anti) symmetrization occurs in $N_{\alpha\mu\nu}$ which in turn means that the latter equation can only give certain parts of $N_{\alpha\mu\nu}$ and not the full tensor. We see therefore that non-degeneracy is essential in order to obtain the full tensor N. Below we give a trivial example where such a degeneracy occurs.
		\section{Examples}
		\textbf{Example 1.}

		Consider the equation
		\beq
		N_{\alpha\mu\nu}-N_{\alpha\mu\nu}=B_{\alpha\mu\nu}
		\eeq
		Here we have $a_{1}=-a_{3}=1$ and $a_{2}=a_{4}=a_{5}=a_{6}=0$ and as a result $\sigma_{1}=0$ meaning that the matrix $A$ is singular and as a result equation ($\ref{eq1}$)  is not enough to specify all the components $N_{\alpha\mu\nu}$.
		Of course in this example the incapability of ($\ref{eq1}$) to fully  specify all the components $N$ of was obvious since we could write the above as $2 N_{[\alpha\mu]\nu}=B_{\alpha\mu\nu}$ meaning that only the antisymmetric part in the first indices of $N$ can be obtained. In more complicated cases, however, it would be quite difficult to spot certain symmetrizations that might occur, especially for the generalized version ($\ref{eq2}$). In these cases the determinant criterion would be of great use in determining whether the given tensor equation can give the components of the full $N$ tensor or not.
		 
				\textbf{Example 2.} Let us now apply the result of our first Theorem in a trivial example one encounters in introductory courses of tensor calculus. There, the metric compatibility condition implies
				\beq
			\Gamma_{\nu\mu\alpha}+	\Gamma_{\mu\nu\alpha}= \partial_{\alpha}g_{\mu\nu}	\;\;, \;\; where \;\;	\Gamma_{\nu\mu\alpha}:=g_{\lambda\nu}\Gamma^{\lambda}_{\;\;\;\mu\nu} \label{LC}
				\eeq
		which along with the torsionlessness of the connection $\Gamma^{\lambda}_{\;\;\;\mu\nu}=\Gamma^{\lambda}_{\;\;\;\nu\mu}$ give us the usual Levi-Civita form of the connection. Recall that the trick to solve for $\Gamma^{\lambda}_{\;\;\;\mu\nu}$ there was to consider two subsequent cyclic permutations of the ($\ref{LC}$) and subtract them from the latter. Let us reproduce this result here by applying the Theorem $1$. In this case, as we have already mentioned the fact that $	\Gamma_{\nu\mu\alpha}$ is symmetric in its last two indices, reduces $A$ to a $3 \times 3$ matrix and we might as well set $a_{4}=a_{5}=a_{6}=0$. Now from $(\ref{LC})$ we read off the coefficients $a_{1}=0$, $a_{2}=a_{3}=1$ and as a result
			\begin{equation}
		A = 
		\begin{pmatrix}
		0 & 1 & 1 \\
		1 & 1 & 0 \\
		1 & 0 & 1  
		\end{pmatrix}
		\end{equation}
		From which we see that $\det(A)=-2 \neq 0$ and we straightforwardly calculate the first row elements of the inverse matrix to be $\tilde{a}_{11}=-1/2$, $\tilde{a}_{12}=\tilde{a}_{13}=1/2$. Then substituting these into ($\ref{co}$) and with the identifications\footnote{Recall that our result holds true not only for tensor but also for tensor densities and connection coefficients as well.} $N_{\alpha\mu\nu}=\Gamma_{\alpha\mu\nu}$ and $B_{\alpha\mu\nu}$
		we arrive at the well known result
		\beq
		\Gamma_{\alpha\mu\nu}=\frac{1}{2}\Big( -\partial_{\alpha}g_{\mu\nu}+\partial_{\nu}g_{\mu\alpha}+\partial_{\mu}g_{\nu\alpha} \Big)
		\eeq
		for the Levi-Civita connection. Of course in this case the same result can be obtained trivially by the classical method we mentioned above. Our intention with these examples here is to illustrate how our general method works. Probably the most useful application of our Theorem to physical systems is the analysis of the connection field equations in Metric-Affine Gravity. There our method will be proven to be all essential in solving for torsion and non-metricity in terms of their sources.

		\bibliographystyle{unsrt}
		\bibliography{ref}

\begin{thebibliography}{10}

\bibitem{iosifidis2019metric}
Damianos Iosifidis.
\newblock Metric-affine gravity and cosmology/aspects of torsion and
  non-metricity in gravity theories.
\newblock {\em arXiv preprint arXiv:1902.09643}, 2019.

\bibitem{iosifidis2019exactly}
Damianos Iosifidis.
\newblock Exactly solvable connections in metric-affine gravity.
\newblock {\em Classical and Quantum Gravity}, 36(8):085001, 2019.

\bibitem{hehl1995metric}
Friedrich~W Hehl, J~Dermott McCrea, Eckehard~W Mielke, and Yuval Ne'eman.
\newblock Metric-affine gauge theory of gravity: field equations, noether
  identities, world spinors, and breaking of dilation invariance.
\newblock {\em Physics Reports}, 258(1-2):1--171, 1995.

\bibitem{qi2018third}
Liqun Qi, Haibin Chen, and Yannan Chen.
\newblock Third order tensors in physics and mechanics.
\newblock In {\em Tensor eigenvalues and their applications}, pages 207--248.
  Springer, 2018.

\bibitem{schouten1954ricci}
JA~Schouten.
\newblock Ricci-calculus. an introduction to tensor analysis and its
  geometrical applications, 1954.

\bibitem{eisenhart2012non}
Luther~Pfahler Eisenhart.
\newblock {\em Non-riemannian geometry}.
\newblock Courier Corporation, 2012.

\bibitem{itin2020decomposition}
Yakov Itin and Shulamit Reches.
\newblock Decomposition of third-order constitutive tensors.
\newblock {\em Mathematics and Mechanics of Solids}, page 10812865211016530,
  2020.

\bibitem{auffray2013geometrical}
Nicolas Auffray.
\newblock Geometrical picture of third-order tensors.
\newblock In {\em Generalized continua as models for materials}, pages 17--40.
  Springer, 2013.

\bibitem{landsberg2012tensors}
Joseph~M Landsberg.
\newblock Tensors: geometry and applications.
\newblock {\em Representation theory}, 381(402):3, 2012.

\bibitem{wolfram1991mathematica}
Stephen Wolfram.
\newblock {\em Mathematica: a system for doing mathematics by computer}.
\newblock Addison Wesley Longman Publishing Co., Inc., 1991.

\end{thebibliography}

	\end{document}